\documentclass[11pt]{amsart}
\usepackage{amsmath}
\usepackage{amsthm}
\usepackage{amssymb}
\usepackage{graphics}
\usepackage{bm}
\usepackage{setspace}
\usepackage{fullpage}

\usepackage{amsfonts}
\usepackage{latexsym}
\usepackage{mathdots}
\usepackage{mathrsfs}
\usepackage{enumitem}
\usepackage{eucal}
\usepackage{color}

\newcommand{\be}{\begin{equation}}
\newcommand{\ee}{\end{equation}}
\newcommand{\ba}{\begin{eqnarray}}
\newcommand{\ea}{\end{eqnarray}}
\newcommand{\bi}{\begin{itemize}}
\newcommand{\ei}{\end{itemize}}
\newcommand{\bn}{\begin{enumerate}}
\newcommand{\en}{\end{enumerate}}
\newcommand{\bbm}{\begin{bmatrix}}
\newcommand{\ebm}{\end{bmatrix}}
\newcommand{\bp}{\begin{proof}}
\newcommand{\ep}{\end{proof}}
\newcommand{\nn}{\nonumber}

\newcommand{\mr}{\ensuremath{\mathrm}}

\newcommand{\mc}{\ensuremath{\mathcal}}

\newcommand{\ov}{\ensuremath{\overline}}
\newcommand{\sm}{\ensuremath{\setminus}}

\newcommand{\ga}{\ensuremath{\gamma}}
\newcommand{\Om}{\ensuremath{\Omega}}

\newcommand{\La}{\ensuremath{\Lambda }}

\newcommand{\om}{\ensuremath{\omega}}

\def\C{\mathbb{C}}

\def\Z{\mathbb{Z}}
\def\N{\mathbb{N}}
\def\B{\mathcal{B}}

\renewcommand{\H}{\ensuremath{\mathcal{H} }}

\newcommand{\K}{\ensuremath{\mathcal{K} }}
\renewcommand{\L}{\ensuremath{\mathscr{L} }}

\newcommand{\A}{\ensuremath{\mathcal{A} }}

\newcommand{\tr}{\ensuremath{\mathrm{tr} }}
\newcommand{\trn}{\ensuremath{\mathrm{tr} _n }}

\newcommand{\tra}{\ensuremath{\mathrm{tr} _{\mathcal{A}} }}
\newcommand{\ptr}{\ensuremath{\mathrm{tr} _{\mathcal{B}} ^{\mathcal{A}} }}
\newcommand{\cc}{\ensuremath{\overline{\mathrm{cc}}  }}

\newcommand{\ip}[2]{\ensuremath{\langle {#1} , {#2} \rangle}}

\newcommand{\ran}[1]{\ensuremath{\mathrm{Ran} \left( {#1} \right) }}
\renewcommand{\ker}[1]{\ensuremath{\mathrm{Ker} ({#1}) }}

\numberwithin{equation}{section}

\numberwithin{subsection}{section}

\newtheorem{thm}[section]{Theorem}
\newtheorem*{thm*}{Theorem}

\newtheorem{lemma}[section]{Lemma}

\newtheorem{cor}[section]{Corollary}

\theoremstyle{definition}
\newtheorem{defn}[section]{Definition}
\newtheorem{remark}[section]{Remark}
\newtheorem{eg}[section]{Example}

\title{Matrix $N-$dilations of quantum channels}

\author{J. Levick}
\address{University of Cape Town}
\email{jlevick@gmail.com}

\author{R.T.W. Martin}
\address{University of Cape Town}
\email{rtwmartin@gmail.com}

\begin{document}
\maketitle
\onehalfspace

\begin{abstract}
We study unital quantum channels which are obtained via partial trace of a $*$-automorphism of a finite unital matrix $*$-algebra.  We prove that any such channel, $q$, on a unital matrix $*$-algebra, $\A$, admits a finite matrix $N-$dilation, $\alpha _N$, for any $N \in \N$. Namely, $\alpha _N$ is a $*$-automorphism of a larger bi-partite matrix algebra $\A \otimes \B$ so that partial trace of $M$-fold self-compositions of $\alpha _N$ yield the $M$-fold self-compositions of the original quantum channel, for any $1\leq M \leq N$. This demonstrates that repeated applications of the channel can be viewed as $*$-automorphic time evolution of a larger finite quantum system.
\end{abstract}

\section{Introduction}

A unital quantum channel $q$ is a unital, completely positive, trace-preserving (UCPTP) linear map. Quantum channels are fundamental objects in quantum computing and quantum information theory where they naturally describe the time evolution of open quantum systems \cite[Chapter 8.2]{NC}.  Given an open quantum system, $\A$, interacting with the quantum system of the environment $\mc{E}$, the total quantum system $$ \A \otimes \mc{E}, $$ is a closed quantum system, and its time evolution is necessarily unitary, \emph{i.e.} described by a (unitary) $*$-automorphism of $\A \otimes \mc{E}$ \cite[Chapter 2.2.2]{NC}. In this paper, we consider quantum channels acting on a finite quantum system, \emph{i.e.} a unital matrix $*$-algebra, $\A$, such that $q$ can be realized as the partial trace of a $*$-automorphism of a larger finite quantum system $\A \otimes \mc{B}$:
$$ q ( A ) \otimes I_{\mc{B}} = \left( \mr{id} _{\A} \otimes I _{\mc{B} } \tr _{\mc{B}} \right) \circ \alpha (A \otimes I_{\mc{B}}) ; \quad \quad  \alpha \in \mr{Aut} (\A \otimes \B).$$ (In the above, and throughout, by a unital matrix $*$-algebra, we mean a unital, self-adjoint matrix algebra.) Any quantum channel with this property is \emph{factorizable} in the sense of \cite{CAD-ergodic,Hup-factor}, and we will review the general definition of factorizability and its relationship to dilation theory in the upcoming Subsection \ref{DilFact}. 

Our main result is that if a unital quantum channel has this property, then for any $N \in \N$, one can construct $*$-automorphisms, $\alpha _N $ of larger finite quantum systems (finite unital matrix $*$-algebras) $\A \otimes \B _N$ so that for any $1 \leq M \leq N$,
$$ q^{(M)} (A) \otimes I_{\B _N} := ( \underbrace{ q \circ q \circ \cdots \circ q }_{M \ \mbox{times}}) (A) \otimes I_{\B _N} = \left( (\mr{id} _\A \otimes I_{\B _N} \tr _{\B_N} )  \circ \alpha _N ^{(M)} \right)  \left(A \otimes I_{\B _N } \right). $$ A physical interpretation is that time evolution of any state or \emph{density matrix} (any positive matrix with unit trace) $(A \otimes I_{\B _N } )$ in the original quantum system $\A \otimes I_{\B _N}$, under the quantum channel, $q$, is implemented by $*$-automorphic time evolution of the larger (finite) quantum system $\A \otimes \B$, up to $M$ `time steps'.

\subsection{Preliminaries}

Let $QC (\A)$ denote the convex set of all unital quantum channels on $\A$, a finite unital matrix $*$-algebra, $\A := \bigoplus _{k=1} ^N \C ^{n _k \times n_k}$, where $\C ^{n \times n}$ denotes the $n\times n$ complex matrices. The normalized trace on $\C ^{n\times n}$ will be denoted by $\trn$ and we assume that there is a fixed, unital faithful trace $\tr _\A$ so that $q$ is $\tr _\A$-preserving: $\tr_\A \circ q = \tr _\A$. If $\{ q_k \} _{k=1} ^p \subset \A$, are a set of contractions so that
$$ q(A) = \sum _{k =1} ^p q_k A q_k ^*, $$ then the $\{ q_k \}$ are called \emph{quantum effects} or \emph{Kraus operators} for $q$ and we write
$$ q \sim \{ q_k \}. $$ It is an easy consequence of Choi's theorem that any $q \in QC (\A)$ is implemented by a set of quantum effects, $q \sim \{ q_k \}$ \cite{Choi}. Moreover, if $\{ q _k \} _{k=1} ^p$, and $\{ q_j ' \} _{j=1} ^{p'}$ are two sets of Kraus operators for $q$, then there is a unitary $U \in \C ^{m \times m}$,
$m = \max \{p, p' \}$, so that (assume without loss of generality that $p <p'$):
$$ U \bbm q _1 \\ \vdots \\ q_p \\ 0 \\ \vdots \\ 0 \ebm = \bbm q_1 ' \\ \vdots \\ q'_{p'} \ebm. $$ We assume throughout, without loss in generality, that $q \sim \{ q_k \} _{k=1} ^p$ is implemented by a linearly independent set of Kraus operators so that $p \in \N$ is minimal. Any unital quantum channel, $q \in QC (\A )$ has a \emph{tracial dual}: $q^\dag \in QC (\A)$. The fact that $q \in QC (\A )$ is unital and trace-preserving implies:
$$ \sum _{k=1} ^p q_k q_k ^* = I _\A ; \quad \quad q \ \mbox{is unital}, $$ and
$$ \sum _{k=1} ^p q_k ^* q_k = I _\A ; \quad \quad q \ \mbox{is trace-preserving}. $$
The tracial dual channel $q^\dag \in QC (\A )$ is defined by
$$ \tra \left( A_1 q^\dag (A_2 ) \right) := \tra \left( q (A_1 ) A _2 \right).$$ This is again a unital quantum channel with effects $q^\dag \sim \{ q _k ^* \}$.

Given $q  \in QC (\A)$ with $q \sim \{ q_k \} _{k=1} ^p$, we say that $U$ is a \emph{unitary matrix factorization} of $q$ if there are $p$ matrices $\{ v_k \} _{k=1} ^p $ so that
$$ U := \sum _{k=1} ^p q_k \otimes v_k; \quad \quad \mbox{is unitary}, $$ and the $v_k$ are trace-orthogonal. That is, we can assume without loss in generality that the $v_k$ all belong to a finite unital matrix $*$-algebra, $\B$, and,
$$ \tr _\B (v_k ^* v_j ) = \delta _{k,j}. $$ If $q \in QC (\A )$ has a unitary matrix factorization, we will say that $q$ is \emph{matrix factorizable}. 
\begin{remark} We note that the set of $q\in QC(\A)$ that are matrix factorizable is closed under convex combinations as well as under composition; a fact that can be proved with only slight modifications to the proof of \cite[Proposition 2]{ricard}: convex combinations of matrix factorizable channels are factorizable by matrix algebras that are direct sums of the respective algebras, and compositions of matrix factorizable channels are factorizable by matrix algebras that are tensor products of the respective algebras.
\end{remark}

\begin{defn} \label{Ndilation}
We say that a $*$-automorphism, $\alpha$, of $ \A \otimes \B$, where $\A , \B$ are finite unital matrix $*$-algebras with fixed faithful, unital traces $\tr _\A, \tr _\B$, is a \emph{$*$-automorphic matrix $N$-dilation} or (more simply) a \emph{matrix $N$-dilation} of $q \in QC (\A )$ if, for any $1 \leq M \leq N$,
$$ q^{(M)} (A) \otimes I _\B  =  \left( \Phi _\B ^\A \circ \alpha ^{(M)} \right) (A \otimes I_\B ) \quad \quad \forall A \in \A. $$
\end{defn}

In the above, $\Phi _\B ^\A$ denotes the unique $\tr _{\A \otimes \B } := \tr _{\A} \otimes \tr _{\B}$-preserving conditional expectation onto the unital $*$-subalgebra $\A \otimes I_\B$. Namely, $\tr _{\A \otimes \B}$ is defined by
$$ \tr _{\A \otimes \B } (A \otimes B ) := \tr_\A (A) \cdot \tr_\B (B); \quad \quad A \in \A, B \in \B, $$ the \emph{partial trace} of $\A \otimes \B$ onto $\A$ is defined as the map $\tr _\B ^\A = \mr{id} _\A \otimes \tr _\B : \A \otimes \B \rightarrow \A$:
$$ \tr _\B ^\A (A \otimes B) = \left( \mr{id} _\A \otimes \tr _\B \right)  (A \otimes B ) := A \cdot \tr _\B (B) \in \A, $$ and the unital $\tr _{\A \otimes \B }$-preserving conditional expectation $\Phi _\B ^\A$ is then
$$ \Phi _\B ^\A := (\mr{id} _\A \otimes \tr _\B ) ^\dag (\mr{id} _\A \otimes \tr _B ) : \A \otimes \B \rightarrow \A \otimes I _\B, $$
$$ \Phi _\B ^\A \left( A \otimes B \right) = A \cdot \tr _\B (B) \otimes I_\B, $$ where $\dag$ denotes tracial dual.

\begin{lemma} \label{onedilation}
 If $\A, \B$ are finite unital matrix $*$-algebras and $q \in QC (\A)$, then $U \in \A \otimes \B$ is a unitary matrix factorization of $q$ if and only if $\mr{Ad} _U$ is a matrix $1$-dilation of $q$.
\end{lemma}

In the above, $\mr{Ad} _U$ denotes the unitary $*$-automorphism of \emph{adjunction by U}: $$\mr{Ad} _U (A\otimes B) := U (A \otimes B) U^*, $$ for some unitary $U \in \A \otimes \B$.

\begin{proof}
    This is a special case of \cite[Theorem 2.2]{Hup-factor}, and easily verified.
\end{proof}

Observe that if $q$ admits a unitary matrix factorization, then the $v_k$ can all be chosen to be matrix units for the unital matrix algebra $\B$. Indeed, if the matrix units for $\B$ are $\{E_k \} _{k=1} ^m$, and $U = \sum _{k=1} ^p q_k \otimes v_k$ is a unitary matrix factorization of $q$, then expanding the $v_k$ in the trace-orthogonal basis $\{ E_k \}$ yields:
$$ v_k = \sum _{j=1} ^m W_{k,j} E_j, $$ and so
\ba U & = & \sum _{k=1} ^p q_k \otimes v_k \nn \\
& = & \sum _{k=1} ^p \sum _{j=1} ^m  W_{k,j} q_k \otimes E_j \nn \\
& = & \sum _{j=1} ^m  \left( \sum _{k=1} ^p c_{k,j} q_k \right) \otimes E_j \nn \\
& =: & \sum _{j=1} ^m q_j ' \otimes E_j. \nn \ea
Since the partial trace of $\mr{Ad} _U$ onto the unital subalgebra $\A \otimes I_\B$ yields $q$, it follows that the $\{q _j ' \} _{j=1} ^m \sim q$ are necessarily another set of effects or Kraus operators for $q$.

\begin{remark}
It follows that all matrix factorizable quantum channels $q \in QC (n):= QC (\C ^{n\times n})$ can be constructed as follows: Fix a unitary matrix $U \in \C ^{nm \times nm} \simeq \C ^{n\times n} \otimes \C ^{m\times m}$, and decompose $U$ into $m^2$ blocks of size $n \times n$:
$$ U = \sum _{k,j =1} ^m q_{k,j} \otimes E_{k,j}; \quad \quad q_{k,j} \in \C ^{n\times n}, $$ where the $E_{k,j}$ are matrix units for $\C ^{m \times m}$. Applying the partial trace of $\C ^{nm \times nm} \simeq \C ^{n\times n} \otimes \C ^{m \times m}$ onto $\C ^{n\times n} \otimes I_m$ to $\mr{Ad} _U$ then yields a unital quantum channel $q \in QC (n)$ so that $ q \sim \{ q_{k,j} \} _{k,j =1} ^m$. The $q_{k,j} \in \C ^{n \times n}$ are a set of $m^2$ effects for $q$, $q \sim \{ q _{k,j} \} _{k,j = 1} ^m$ and $U$ is then a unitary matrix factorization of $q$. Equivalently, by Lemma \ref{onedilation}, $\mr{Ad} _U$ is a matrix $1$-dilation of $q$.
\end{remark}

\begin{eg}{ (Discrete Fourier Transform)}

Let $\om := e^{-i\frac{2\pi}{N}}$ be a primitive $N^{th}$ root of unity, and set
$$[\Om _{kj} ] := \frac{1}{\sqrt{N}} [ \om ^{(k-1)(j-1)}] \in \C ^{N \times N}; \quad \quad 1 \leq k,j, \leq N. $$
That is,
$$ \Om = \frac{1}{\sqrt{N}} \bbm 1 & 1 & 1 & 1 & \cdots & 1  \\
1 & \om & \om ^2 & \om ^3 & \cdots & \om ^{N-1} \\
1 & \om ^2 & \om ^4 & \om ^ 6 &  \cdots & \om ^{2(N-1)} \\
1 & \om ^3 & \om ^6 & \om ^9 & \cdots & \om ^{3(N-1)} \\
\vdots & & & & \ddots & \vdots \\
1 & \om ^{N-1} & \om ^{2(N-1)} & \cdots & & \om ^{(N-1) ^2} \ebm. $$ This is the $N$-point Discrete Fourier Transform (DFT) matrix, and it is unitary.
Given any $n \in \N$, if we choose $N := n \cdot m$, for some $m \in \N$, then we can break $\Om$ into $m^2$ blocks of matrices in $\C ^{n\times n}$ as:
$$ \Om = \sum _{k,j}  \Om _{k,j} \otimes E_{k,j}; $$ where the $E_{k,j}$ are the standard matrix units for $\C ^{m\times m}$.
It follows that the unital quantum channel $q \in QC (n)$ defined by $q \sim \{ \Om _{k,j} \}$ has $\Om$ as a unitary matrix factorization.
\end{eg}

\begin{eg}{ (Random unitary channels)}
Choose numbers $\{ p_1, ... , p_N \} >0$ so that $$ \sum _{k=1} ^N p_k = 1, $$ and let $\{ U_k \} _{k=1} ^p \in \C ^{n\times n}$ be any unitary matrices.
Let $\B := \C^N$ denote the diagonal $*$-algebra of $N\times N$ diagonal matrices with faithful normal tracial state:
$$ \tr _\B ( \mr{diag} (b_1 , ... , b_N ) ) := \sum _{k=1} ^N p_k b_k, $$ and let $\{ E _k \} _{k=1} ^N$ be diagonal matrix units for this algebra.
It is easy to check that
$$ q(A) := \sum _{k=1} ^N p_k U_k A U_k^*, $$ is a unital quantum channel on $\A = \C ^{n\times n}$, and that
$$ V := \sum _{k=1} ^N U_k \otimes E_k; \quad \quad [V] = \bbm U_1 & &  & \\
 & U_2 & & \\ & & \ddots &  \\ & & & U_N\ebm$$ is a unitary matrix factorization of $q$. (In the above blank entries are all zero, and $[V]$ denotes the matrix representation of $V$ in the canonical basis.)
\end{eg}

\begin{eg}{ (Schur product channels)} \label{Schureg}
The Schur product of two matrices $A$ and $B$ whose dimensions are the same, denoted $A\circ B$, is simply their entry-wise product: $(A\circ B)_{ij} = A_{ij}B_{ij}$. A Schur product channel is a quantum channel $q:\C^{n\times n}\rightarrow \C^{n\times n}$ that outputs the Schur product of each input matrix with some fixed output matrix:

$$q(X) = X\circ C$$ for some fixed $C$.
A channel of this form is completely positive whenever $C \geq 0$; to see this, note that for $v\in \C^n$ and $D_v = \mathrm{diag}(v)$, $X\circ (vv^*)=D_v X D_v^*$. If $C\geq 0$, then there exist $\{v_i\}_{i=1}^p$ such that $C= \sum_{i=1}^p v_iv_i^*$, and so
$$q(X) = X\circ C = X \circ (\sum_{i=1}^p v_iv_i^*) = \sum_{i=1}^p D_{v_i}X D_{v_i}^*;$$
hence $q \sim \{D_{v_i}\}_{i=1}^p$ is an operator-sum form for $q$. In general, Schur product channels correspond to channels whose Kraus operators are diagonal.\\
If $C\geq 0$, the map $q(X) = X\circ C$ is trace-preserving if and only if $\sum_{i=1}^n x_{ii}c_{ii} = \sum_{i=1}^n x_{ii}$ for all $X$; i.e., if and only if $c_{ii} = 1$. Thus, $C$ is a correlation matrix: a positive semidefinite matrix with $1$'s down the diagonal. In this case, $q$ is automatically unital as well as trace-preserving.\\
From the above, it is clear that all matrix-factorizable Schur product channels arise in the following way:
$$q(X) = \tr_B U(X\otimes I)U^*$$ for some unitary $U\in \C^{n\times n}\otimes \mathcal{N}$. If we write $U = \sum_{k} u_k\otimes E_k$ where $E_k$ are matrix units for $\mathcal{N}$, then for $q$ to be a Schur product channel, $u_k$ must all be diagonal.\\
Factorizability of Schur product channels is intimately connected to the geometry of the convex set $\mathcal{E}_n:=\{C \in \C^{n\times n}: C\geq 0, c_{ii} =1\}$. First of all, rank-one correlation matrices induce unitary Schur product channels. This is because, if $C = vv^*$, and $c_{ii}=1$, necessarily $|v_i|=1$; hence $q(X) = X\circ C = D_vX D_v^*$ where $D_v$ is a diagonal matrix whose diagonal entries all have modulus-$1$, so $D_v$ is unitary.\\
This means that if a correlation matrix $C$ is in the convex hull of rank-one correlation matrices, the channel $q(C)$ is random unitary. For $n=1,2,3$ rank-one correlation matrices are the only extreme points of the set $\mathcal{E}_n$, and so all Schur product channels on the $n\times n$ matrices for $n\leq 3$ are random unitary, hence factorizable.\\
For $n\geq 4$, there are rank-$k$ extreme points of $\mathcal{E}_n$ for all $k\leq \sqrt{n}$ \cite{Li-xtrm}. If $C$ is an extreme point of rank $\geq 2$ of $\mathcal{E}_n$, then the channel $q(X) = X\circ C$ is not factorizable \cite{Hup-factor} \cite{dykema}. \\
It is also possible to have a Schur product channel that is factorizable, but not random unitary (i.e., $C$ is not in the convex hull of rank-$1$ correlation matrices). Haagerup and Musat exhibit such an example for $n=6$.\\
More generally, Haagerup and Musat have shown that a correlation matrix $C=(c_{ij})\in \mathcal{E}_n$ is factorizable if and only if there exists a finite von Neumann algebra $\mathcal{N}$ and unitaries $\{U_i\}_{i=1}^n$ in $\mathcal{N}$ such that $\mathrm{Tr}_{\mathcal{N}}(U_i^*U_j) = c_{ij}$. The question of whether the closure of the set of matrix factorizable correlation matrices is the same as the set of factorizable correlation matrices is equivalent to Connes' embedding conjecture \cite{dykema}. \\
One wide class of Schur product channels that are known to be factorizable is the set of Schur product channels arising from correlation matrices with all real entries \cite{ricard}\cite{Hup-factor}. Such a channel always admits a factorization by means of trace-orthogonal, anti-commuting unitaries. \\

\end{eg}

\subsection{Dilation and Factorization} \label{DilFact}

The notions of dilation and factorization of quantum channels were originally introduced in the more general context of what are called \emph{Markov maps} between finite von Neumann algebras. If $(\A , \phi ); \ (\B , \psi )$ are two finite von Neumann algebras equipped with faithful normal states, $\phi , \psi$, a unital completely positive map $q : \A \rightarrow \B$ is called a $(\phi , \psi)$-Markov map if $\psi \circ q = \phi$, and if $q$ intertwines the one-parameter $*$-automorphism groups of $\phi $ and $\psi$ \cite{CAD-ergodic,Hup-factor}. These conditions ensure that any $(\phi , \psi)-$Markov map has a dual $(\psi, \phi )$-Markov map $q^\dag : \B \rightarrow \A$ defined by
$$ \phi (q^\dag (B) A ) := \psi (B q(A) ); \quad \quad A \in \A, B \in \B. $$ A quantum channel is the special case where both $\A ,\B$ are finite unital matrix $*$-algebras and $\phi, \psi$ are unital faithful, normalized traces. Recall we consider the case of $q \in QC (\A)$ so that $q: \A \rightarrow \A$ is a unital quantum channel on a unital matrix $*$-algebra $\A$.

The original definition of factorizability of a $(\phi ,\psi)$-Markov map $q$ is: $q : (\A , \phi ) \rightarrow (\B ,\psi )$ is \emph{factorizable} if
there exists a pair, $(\mc{C} , \ga)$, consisting of a finite von Neumann algebra $\mc{C}$ equipped with faithful normal state $\ga$, $*$-monomorphisms $\alpha : \A \rightarrow \mc{C}$ and $\beta : \B \rightarrow \mc{C}$ which are $(\phi, \ga)$, and $(\psi, \ga)$-Markov, respectively, so that
$$ q = \beta ^\dag \circ \alpha, $$ \cite{CAD-ergodic,Hup-factor}. For the case of interest to us, unital quantum channels $q : \A \rightarrow \B$, between finite unital matrix $*$-algebras, $\A, \B$, $q$ has a factorization in this sense with $\mc{C}$ another finite unital matrix $*$-algebra if and only if $q$ has a unitary matrix factorization (or equivalently, $q$ has a matrix $1$-dilation) as in Lemma \ref{onedilation} above, and we will not have need for this fully general definition of factorization.

The concept of dilation of a $(\phi , \phi)$-Markov map $q: (\A , \phi ) \rightarrow (\A , \phi)$ was originally introduced by K\"{u}mmerer in \cite[Definition 2.1.1]{Kum-Markov}:
A dilation of $q$ is a quadruple $(\mc{M} , \ga , \alpha , J )$ where $\mc{M}$ is a finite von Neumann algebra with faithful normal state $\ga$, $\alpha$ is a $\ga$-preserving $*$-automorphism of $\mc{M}$, and $J : \A \rightarrow \mc{M}$ is a $(\phi, \ga)$-Markov $*$-monomorphism so that
$$ q ^{(n)} = J ^* \alpha ^{(n)} J, $$ and as before, $q^{(n)}$, denotes $n$-fold self-composition. As proven in \cite[Theorem 4.4]{Hup-factor}, a $\phi$-Markov map $q: (\A , \phi ) \rightarrow (\A , \phi)$ has a dilation if and only if it is factorizable.

We focus on the special case of unital quantum channels $q \in QC (\A)$ on finite unital matrix algebras $\A$, the set of all $\tr _\A$-Markov maps. In this setting, if $q$ is factorizable, we refer to a dilation $(\mc{M}, \ga ,\alpha , J )$, of $q$, as defined above as a \emph{power dilation}. Corollary \ref{nopower} shows that if $\ga$ is a faithful trace, $\alpha$ a $*$-automorphism, and $J$ is an embedding of $\A$ into $\mc{M}$, then $\mc{M}$ cannot be a (finite) type-I von Neumann algebra. This motivates the consideration of matrix $N$-dilations of quantum channels as defined in Definition \ref{Ndilation}. Matrix $N$-dilations may also be of more interest in quantum information and quantum computing since they act on finite quantum systems, and hence are, in principle, easier to implement in an experimental setting, \emph{e.g} in a quantum computer \cite{NC}.  We will show that a unital quantum channel $q \in QC (\A )$ has a matrix $N$-dilation, $\alpha _N$, for any $N \in \N$, if and only if $q$ has a unitary matrix factorization, and our explicit and simple construction provides upper bounds on the dimension of the unital matrix algebra on which $\alpha _N$ acts:

\begin{thm*}{ (Theorem \ref{main})}
Let $\A$ be a unital matrix $*$-algebra. A unital quantum channel $q \in QC (\A )$, is factorizable with matrix $1$-dilation $\alpha _1$ acting on $\A \otimes \B$, if and only if
$q$ has a matrix $N$-dilation $\alpha _N$ acting on $\A \otimes \B ^{\otimes N}$, for any $N \in \N$.
\end{thm*}
In the above $\B ^{\otimes N} := \underbrace{\B \otimes \B \otimes \cdots \otimes \B }_{N \ \mbox{times}}.$

\section{$N-$dilations of quantum channels}

Let $q \in QC (\A )$ be a unital quantum channel on a finite unital matrix $*$-algebra, $\A$, with faithful normalized trace $\tra$, and Kraus operators $q \sim \{ q_k \} _{k=1} ^p$.  Assume that $q$ has a unitary matrix factorization:
\be U := \sum _{k=1} ^p q_k \otimes b_k \in \A \otimes \B. \label{unifact} \ee
For simplicity of notation, let
$$ \tr _\B ^\A := \mr{id} _\A \otimes \tr _\B, $$ be the partial trace of $\A \otimes \B$ onto $\A$, and let
$$ \Phi := \mr{id} _\A \otimes I _\B \cdot \tr _\B := (\tr _\B ^\A ) ^\dag \tr _\B ^\A, $$ be the unique, unital, $\tr _\A \otimes \tr _\B$-preserving conditional expectation of $\A \otimes \B$ onto $\A \otimes I_\B$. Recall that since $U$ is a unitary matrix factorization of $q$, $\mr{Ad} _U$ is a matrix $1-$dilation of $q$, by Lemma \ref{onedilation}:
$$ (q \otimes I_\B ) (A \otimes I_\B ) = \left( \mr{id} _\A \otimes \tr _\B \right) \left( U ( A \otimes I_\B ) U^* \right); \quad \quad A \in \A. $$ Equivalently,
$$ q \circ \Phi = \tr _\B ^\A \circ \mr{Ad} _U \circ \Phi, \quad \quad \mbox{or} \quad  \Phi \circ \mr{Ad} _U \circ \Phi  = q \circ \tr _\B ^\A \otimes I_\B. $$

For any $N \in \N$, consider the unital matrix $*$-algebra $$ \A \otimes \underbrace{\B \otimes \B \otimes \cdots \otimes \B}_{N \ \mbox{times}} =: \A \otimes \B ^{\otimes N}. $$ Set,
$$ U_N := U \otimes I_{\B ^{\otimes (N-1)}} \in \A \otimes \B ^{\otimes N}.$$  Observe that, by uniqueness of the $\tr _\A \otimes \tr _{\B ^{\otimes N}}$-preserving conditional expectation, $\Phi _N$, of $\A \otimes \B ^{\otimes N}$ onto $\A \otimes I_{\B ^{\otimes N}} $,
\be \Phi _N ( A \otimes B_1 \otimes B_2 \otimes \cdots \otimes B_N )= \tr _\B ^{\A} \left( \tr _\B ^{\A} \left( \cdots \tr _\B ^{\A} \left( \tr _\B ^{\A} (A \otimes B_1 ) \otimes B_2 \right) \cdots \right) \otimes B_N \right) \otimes I_{\B ^{\otimes N}}  \label{ceform} \ee

Also observe that
\ba \tr _\B ^\A \circ (q \otimes \mr{id} _\B ) (A \otimes B) & = & \tr _\B ^\A ( q(A) \otimes B ) \nn \\
& = & q(A) \cdot \tr_\B (B), \nn \ea and this proves that
\be \tr _\B ^\A \circ (q \otimes \mr{id} _\B ) = q \circ \tr _\B ^\A. \label{commute} \ee

Define the $*-$automorphism, $\sigma _N$, of $\A \otimes \B ^{\otimes N}$ by a cyclic permutation of the $N$ tensor factors of $\B ^{\otimes N}$: For any $A \in \A$, and $B_1, ... ,B_N \in \B$,
$$ \sigma _N (A \otimes B_1 \otimes B_2 \cdots \otimes B_N) := A \otimes B_N \otimes B_1 \otimes B_2 \cdots \otimes B_{N-1}. $$
Finally, define the $*$-automorphism
\be \alpha _N :=  \mr{Ad} _{U_N} \circ \sigma _N, \label{Ndilform} \ee a $*-$automorphism of $\A \otimes \B ^{\otimes N}$.
\begin{thm} \label{main}
    Let $q \in QC (\A)$ be a matrix factorizable quantum channel, $q(A) = (\mr{id} _\A \otimes \tr _\B ) \circ \mr{Ad} _U (A \otimes I_\B)$ where $q \sim \{ q_k \} _{k=1} ^p$, and $U := \sum _{k=1} ^p q_k \otimes b_k \in \A \otimes \B$ is unitary. Then the   $*$-automorphism, $\alpha _N$, defined above, is a finite matrix $N$-dilation of $q$ acting on the algebra $\A \otimes \B ^{\otimes N}$.
\end{thm}
\begin{proof}
Suppose that $1 \leq M \leq N$. The action of $\alpha ^{(M)}$ on $A \otimes I_{\B ^{\otimes N}}$ is:
$$ \alpha ^{(M)} (A \otimes I_{\B ^{\otimes N}} ) = \sum _{\substack{j_1, ... , j_M =1 \\ k_1, ... , k_M =1} } ^p q_{j_M} \cdots q_{j_1} A q_{k_1} ^* \cdots q_{k_M} ^* \otimes b_{j_M} b_{k_M} ^* \otimes b_{j_{M-1}} b_{k_{M-1}} ^* \otimes \cdots \otimes b_{j_1}b_{k_1} ^* \otimes I_{\B ^{\otimes (N-M)}}. $$

Then, by equation (\ref{ceform}), \scriptsize
\ba & &  \Phi _N \circ \alpha _N ^{(M)}  \left( A \otimes I_{\B ^{\otimes N}} \right) \label{orgexp} \\
& = & \ptr \left( \ptr \left( \cdots \left( \underbrace{\ptr \left( \sum _{\substack{j_1, ..., j_M \\ k_1 , ..., k_M}  =1} ^p q_{j_M} \cdots q_{j_1} A q_{k_1} ^* \cdots q_{k_M} ^* \otimes b_{j_M} b_{k_M} ^* \right)}_{(A)}
\otimes b _{j_{M-1}} b _{k_{M-1} } ^* \right) \cdots \right) \otimes b_{j_1} b_{k_1} ^* \right) \otimes I_{\B ^{\otimes N}} \nn  \ea \normalsize
Since $\mr{Ad} _U$ is a $1$-dilation of $q$, the expression (A) evaluates to:
\ba
\mathrm{(A)} & = & \ptr \circ \mr{Ad} _U \left(   \sum _{\substack{j_1, ..., j_{M-1}  \\  k_1 , ..., k_{M-1}}  =1} ^p q_{j_{M-1}} \cdots q_{j_1} A q_{k_1} ^* \cdots q_{k_{M-1}} ^* \otimes I_\B \right) \nn \\ & = & q \left(   \sum _{\substack{j_1, ..., j_{M-1} \\ k_1 , ..., k_{M-1}}  =1}  q_{j_{M-1}} \cdots q_{j_1} A q_{k_1} ^* \cdots q_{k_{M-1}} ^*  \right). \nn \ea
Applying the conditional expectation formula (\ref{ceform}), and the commutation formula (\ref{commute}), the original expression (\ref{orgexp}) becomes \tiny
\ba & & \ptr \left( \ptr \left( \cdots \left( \tr _\B ^\A \circ (q \otimes \mr{id} _\B )    \left( \sum _{\substack{j_1, ..., j_{M-1} \\ k_1 , ..., k_{M-1}}  =1} ^p q_{j_{M-1}} \cdots q_{j_1} A q_{k_1} ^* \cdots q_{k_{M-1}} ^* \otimes b_{j_{M-1}} b_{k_{M-1}} ^*  \right) \otimes b _{j_{M-2}} b _{k_{M-2} } ^* \right) \cdots \right) \otimes b_{j_1} b_{k_1} ^* \right) \otimes I_{\B ^{\otimes N}} \nn \\
& & =  \ptr  \left( \cdots \left( \tr _\B ^\A \circ (q \otimes \mr{id} _\B ) \circ \mr{Ad _U}     \left( \sum _{\substack{j_1, ..., j_{M-2} \\ k_1 , ..., k_{M-2}}  =1} ^p q_{j_{M-2}} \cdots q_{j_1} A q_{k_1} ^* \cdots q_{k_{M-2}} ^* \otimes I _\B  \right)
\otimes b _{j_{M-2}} b _{k_{M-2} } ^* \right) \cdots \otimes b_{j_1} b_{k_1} ^* \right) \otimes I_{\B ^{\otimes N}} \nn \\
&=& \ptr  \left( \cdots \left( q\circ \tr _\B ^\A   \circ \mr{Ad _U} \circ \Phi _1    \left( \sum _{\substack{j_1, ..., j_{M-2} \\ k_1 , ..., k_{M-2}}  =1} ^p q_{j_{M-2}} \cdots q_{j_1} A q_{k_1} ^* \cdots q_{k_{M-2}} ^* \otimes I _\B  \right)
\otimes b _{j_{M-2}} b _{k_{M-2} } ^* \right) \cdots  \otimes b_{j_1} b_{k_1} ^* \right) \otimes I_{\B ^{\otimes N}} \nn \\
& &\ptr  \left( \cdots \left( q ^{(2)}   \left( \sum _{\substack{j_1, ..., j_{M-2} \\ k_1 , ..., k_{M-2}}  =1} ^p q_{j_{M-2}} \cdots q_{j_1} A q_{k_1} ^* \cdots q_{k_{M-2}} ^* \otimes b _{j_{M-2}} b _{k_{M-2} } ^* \right) \cdots  \otimes b_{j_1} b_{k_1} ^* \right) \right) \otimes I_{\B ^{\otimes N}} \nn \\ &  & \cdots = q ^{(M)} (A) \otimes I_{\B ^{\otimes N}}. \nn \ea
\normalsize
\end{proof}

\begin{remark}
It seems reasonable to expect that one can take the limit of the above construction of the $N-$dilation, $\alpha _N$, in a suitable way to obtain a power dilation, $\alpha$, of the original quantum channel $q \in QC (\A )$, which acts on a finite von Neumann algebra realized as an infinite tensor product of unital matrix $*$-algebras \cite{vN-infinite,AW-factors}.
\end{remark}

\section{Representing contractions and dilations}

This section develops some general characterizations of unital quantum channels which may be of independent interest. Namely, we will define a natural map from unital quantum channels, $q \in QC (\A )$ acting on finite unital matrix $*$-algebras, $\A$, to contractions $T_q$ acting on the Gelfand-Naimark-Segal (GNS) Hilbert space associated to $\A$ and its faithful normalized trace $\tr _\A$.

We will show that many nice properties of quantum channels (including the relationship between a factorizable unital quantum channel and its $N$-dilations) correspond to similarly nice properties under the map $q \mapsto T_q$. We will apply this construction to prove, in particular, that if $q \in QC (\A)$ is factorizable, then $q$ is either already a unitary $*$-automorphism, or any power dilation, $\alpha$ of $q$ cannot act on a finite unital matrix $*$-algebra (Corollary \ref{nopower}).

Let $\A$ be a finite unital matrix $*$-algebra, with faithful normalized trace, $\tra$. Let $L^2 (\A) = L^2 (\A , \tra )$ denote the tracial GNS space with inner product:
$$ \ip{A_1}{A_2} _{\tra} := \tra \left( A_1^*A_2 \right); \quad \quad A_1, A_2 \in \A, $$ (conjugate linear in the first argument). Also let $L : \A \rightarrow \L (L^2 (\A))$ be the left regular representation of $\A$ on $L^2 (\A)$:
$$ L_A A_1 := A A_1; \quad A, A_1 \in \A. $$

\begin{defn}
If $\La : \A \rightarrow \A$ is any linear (or anti-linear) map, we can define a corresponding linear map $T _\La \in \L (L^2 (\A ))$ by the formula:
$$ T_\La A := \La (A) \in L^2 (\A ). $$ We will call the map $T _\La$, the \emph{representing map} of $\La$.
\end{defn}

\begin{lemma} \label{repcontract}
If $q \in QC (\A)$ then the representing map $T _q : L^2 (\A) \rightarrow L^2 (\A )$ defined by $$ T_q ( A) := q(A) ; \quad \quad A \in \A,$$
is a contraction.
\end{lemma}

\begin{proof}
    This is an easy application of the Schwarz inequality for $2$-positive maps (\cite[Proposition 3.3.]{Paulsen}) and the fact that $q$ preserves the trace:
\ba \| T_q A \|  ^2 & = & \ip{T_q A}{T_q A} _{\tra} \nn \\
& = & \tra \left( q(A) ^* q(A) \right) \nn \\
& \leq & \tra \left( q (A ^* A ) \right) \quad \quad \mbox{(Schwarz inequality)} \nn \\
& = & \tra (A^* A) \quad \quad \mbox{(trace preservation of $q$)} \nn \\
& = & \| A \| _{\tra} ^2. \nn \ea
\end{proof}

\begin{remark} \label{effects}
We call $T_q$ the \emph{representing contraction} of the quantum channel $q$.  Observe that
$$ T_{q^\dag} = T_q ^*, $$ and that if $q_1, q_2$ are quantum channels then
$$ T_{q_2 \circ q_1} = T_{q_2} T_{q_1}. $$ The map $q \mapsto T_q$ takes $QC (\A )$ onto a unital convex subset of the closed unit ball of
$\L (L^2 (\A ))$ which is closed under products (\emph{i.e.} a closed convex monoid).
\end{remark}

\begin{thm} \label{represent}
    Let $\A$ be a finite unital matrix $*$-algebra. Given $q \in QC (\A)$, $q$ is:
\bn
    \item a $*-$automorphism if and only if $T_q$ is unitary.
    \item a $\tra$-preserving conditional expectation onto a unital $*-$subalgebra if and only if $T_q$ is a projection.
     \item a $*-$monomorphism of a unital $*$-subalgebra into $\A$ composed with the conditional expectation onto that unital $*$-subalgebra if and only if $T_q$ is a partial isometry.
\en
\end{thm}

A $*$-monomorphism is an injective $*$-homomorphism. Any quantum channel $q \in QC (\A)$ is necessarily \emph{faithful} since it preserves the faithful normalized trace $\tra$:
$$ q(A^*A) = 0 \quad \Leftrightarrow A =0. $$ It follows that if $q$ is a $*$-homomorphism on all of $\A$, it is necessarily injective. Since everything is finite
dimensional this implies that $T_q$ and hence $q$ is also surjective so that $q$ is a $*$-automorphism.

\begin{lemma} \label{sakern}
    The kernel and co-kernel of any representing contraction, $T_q$, of a unital quantum channel $q \in QC (\A)$ are self-adjoint:
If $A \in \ker{T_q}$ or $\ker{T_q } ^\perp$ respectively, then so is $A^*$.
\end{lemma}
    Applying this fact to the dual channel also shows that the range and orthogonal complement of the range of any quantum channel are self-adjoint.
\begin{proof}
    Clearly $\ker{T_q}$ is self-adjoint since $q(A) = 0 \Leftrightarrow q(A) ^* = q (A^* ) = 0 $. Now suppose that
$$A \in \ker{T_q} ^\perp = \ran{T_q ^*} = \ran{T_{q ^\dag} }. $$ Then there is a $B \in \A$ so that
$q^\dag (B) = A$, and by self-adjointness, $q^\dag (B^* ) = A^*$ so that $A^* \in \ran{T_q ^*} = \ker{T_q} ^\perp$ as well.
\end{proof}

\begin{proof}{ (of Theorem \ref{represent})}
First suppose that $q$ is a $*$-automorphism. Then for any $A \in \A$. Then
\ba \| T_q A \| ^2 & = & \tra (q(A) ^* q(A) ) \nn \\
& = & \tra ( q (A^* A ) \quad \quad \mbox{($q$ is a $*$-homomorphism)} \nn \\
& = & \tra{A^* A} \quad \quad \mbox{($q$ preserves the normalized trace)} \nn \\
& = & \| A \| ^2 _{\tra}, \nn \ea and this proves that $T_q$ is unitary. Conversely if $T_q$ is unitary then
$$ \| T_q A \| ^2 = \| A \| ^2, $$ which implies that
\ba  0 &=& \tra (A^* A -q(A) ^* q(A) ) \nn \\
&=& \tra ( q (A^* A) - q(A) ^* q (A) ). \nn \ea  The argument is non-negative by the Schwarz inequality so that
$q(A^* A) = q(A) ^* q (A)$ by faithfulness of the normalized trace. Choi's theorem on multiplicative domains \cite[Theorem 3.18]{Paulsen} then implies
that $q$ is a $*$-automorphism.

If $q$ is a conditional expectation then
$$ q(q(A)) = q(A), $$ so that $T_q ^2 = T_q$ and $T_q$ is idempotent. Also, for any $A, B \in \A$,
\ba \tra (q^\dag(B) ^* A ) & =&  \tra (B^* q(A) ) \nn \\
    & = & \tra \left( q (B^* q(A) ) \right) \nn \\
& =& \tra (q(B)^* q(A) ) \nn \quad \quad \mbox{(conditional expectation property)} \\
& = & \tra \left( q ( q(B^*) A) \right) \nn \quad \quad \mbox{(same property)} \\
& = & \tra (q(B^* ) A ) \nn \quad \quad \mbox{(trace preservation)}. \nn \ea This proves that $q = q^\dag$ so that
$T_q = T_q ^*$, and $T_q$ is a projection. Conversely if $T_q$ is a projection then $q = q^\dag = q^2$ which implies that $q$ is a conditional expectation.

Finally suppose that $q = \alpha \circ \Phi$ is a $*$-monomorphism, $\alpha$, of a unital $*$-subalgebra $\mc{B} \subset \A$ into $\A$ composed with the $\tra$-preserving conditional expectation $\Phi$ onto that subalgebra. We can identify $L^2 (\mc{B} , \tra )$ as a Hilbert subspace of $L^2 (\A)$, and it follows as in
the first part of the proof that $T _\alpha | _{L^2 (\B , \tra)}$ is an isometry of the subspace $L^2 (\mc{B}, \tra)$ into $L^2 (\A )$, while $T_\Phi$ is the orthogonal projection of $L^2 (\A)$ onto $L^2 (\mc{B} , \tra)$.

Conversely, if $T_q$ is a partial isometry, we can write $T_q = T_q P$ where $P$ is the projection onto $\ker{q} ^\perp$. As in the first part of the
proof, for any $A \in \ker{q} ^\perp$,
$$ q(A^*A) = q (A) ^* q(A). $$ Since $\ker{q} ^\perp$ is self-adjoint (by the previous lemma), we also have
$$ q (A A ^* ) = q (A) q(A) ^*. $$ By Choi's multiplicative domain theorem we have that the set of all $A \in \A$ obeying these two
identities is a unital $*$-subalgebra,$\mc{B}$, of $\A$ so that $\ker{q} ^\perp \subset \mc{B}$. Moreover Choi's multiplicative domain theorem actually shows
that for any $B \in \mc{B}$ and $A \in \A$, $q(A B) = q (A) q(B)$. Suppose that there is a $B \in \mc{B} \sm \ker{q} ^\perp$. Then
it would follow that there is a $C \in \mc{B}$ such that $C \in \ker{q}$. It would then follow that $q(C^* C) = q(C) ^* q(C) =0$, contradicting
the faithfulness of $q$. It follows that $\mc{B} = \ker{q} ^\perp$. Hence $q$ is an injective $*$-homomorphism, $\alpha$, when restricted to the unital $*$-subalgebra $\mc{B} = \ker{q} ^\perp$, and we can write $q = \alpha \circ \Phi$ where $\Phi \in QC (\A)$ is the $\tra$-preserving conditional expectation onto the unital $*$-subalgebra $\ker{q} ^\perp$.
\end{proof}
The previous Theorem and proof are easily modified to apply to the case of unital quantum channels between different unital matrix $*-$algebras.
\begin{remark}
Any contraction, $T$, can be decomposed (uniquely) as $T=V+C$ where $V$ is a partial isometry and $C$ is a strict (assuming $T$ acts on a finite-dimensional Hilbert space) contraction, $\| C \| < 1$, with $\ker{C} ^\perp \subseteq \ker{V}$ and $\ran{C} \subseteq \ran{V} ^\perp$. The defect indices of $T$ are defined to be the dimensions of the (in general, closures of the) ranges of the defect operators $\sqrt{I-T^*T}$ and $\sqrt{I-TT^*}$, and these are equal to the defect indices of the partial isometry $V$. The defect indices of $V$ are simply the dimensions of the defect spaces $\ker{V}$, $\ran{V} ^\perp$ and in this sense measure how close $V$ is to a unitary. If we assume that $T$ acts on a finite dimensional Hilbert space, then $V$ is a partial isometry of a finite dimensional Hilbert space into itself so that $\ker{V}, \ran{V} ^\perp$ must have the same dimension, and the defect indices must be equal.
\end{remark}

\begin{cor}
Let $T=T_q$ be the representing contraction of $q \in QC (\A )$, and $T_q = V+C$ be the above isometric-contractive decomposition of $T$. Then,
\bn
    \item $\ker{V} ^\perp$ is the multiplicative domain of $q$.
    \item $\ker{V} ^\perp \cap \ran{V}$ is the stable multiplicative domain of $q$.
    \item $V=T_{q'}$ is the representing contraction of $q' := q \circ \Phi$, where $\Phi =\Phi _{\mr{Mult} (q)}$ is the unique $\tra$-preserving conditional expectation onto the multiplicative domain of $q$.
    \item The (equal) defect indices of $T_q$ can only take values in the dimensions of unital $*$-subalgebras of $\A$.
\en 
\end{cor}
Recall that the \emph{multiplicative domain} of any $q \in QC (\A ) $ is defined as:
$$ \mr{Mult} (q) := \{ A \in \A | \ q (A B ) = q(A) q(B) \ \mbox{and} \ q(BA) = q(B) q(A) \ \forall B \in \A \}, $$ and the \emph{stable multiplicative domain} is the intersection of the multiplicative domains of $q^{(k)}$, for all $k \in \N$. 
\begin{proof}
If $A \in \ker{V} ^\perp$, then $\| A \| _{\tra} =  \| q (A) \| _{\tra}$, and the argument as in the proof of Theorem \ref{represent} shows that $q(A) ^* q(A) = q(A^* A)$. By Choi's theorem on multiplicative domains, \cite[Theorem 3.18]{Paulsen}, $q(BA) = q(B) q(A)$ for every $B \in \A$ so that $A$ belongs to the right multiplicative domain of $q$. However, by taking adjoints we also have that 
$$ \| A^* \|^2 _{\tra} = \tra ( A A^* ) = \tra (A ^* A ) = \| A \|^2 _{\tra} = \| q(A) \| ^2 _{\tra} = \| q(A) ^* \| ^2 _{\tra}  = \| q (A^* ) \| ^2 _{\tra}. $$ We conclude that $A^*$ is also in the right multiplicative domain, so that the right (and left) multiplicative domain is self-adjoint and coincides with the multiplicative domain. It is not difficult to show that $\ker{V} ^\perp$ is the entire multiplicative domain of $q$ so that $V = T_{q'}$ is the representing contraction of $q' := q \circ \Phi _{\mr{Mult} (q)}$. The remaining items are now clear.
\end{proof}

\subsection{A conjugation commuting with representing contractions}

Let $\{ e_k \} _{k=1} ^n$ be the canonical orthonormal basis for $\C ^n$, and let $\{ E_{j,k} \} _{1 \leq j,k \leq n}$ be the
corresponding matrix units for $\C ^{n\times n}$. We can define a canonical orthogonal basis for $L^2 (\trn)$, $\{ E_k \} _{k =1} ^{n^2}$ by
$$ E_k := \left\{ \begin{array}{cc} E _{1,k} & 1 \leq k \leq n  \\
E_{2,(k-n)} & n+1 \leq k \leq 2n \\
\vdots & \vdots \\
E _{j, (k-n(j-1)) } & (j-1)n +1 \leq k \leq jn \\
\vdots & \vdots \\
E _{n, (k-n(n-1))} & n(n-1) +1  \leq k \leq n^2  \end{array} \right. $$
For $1 \leq k \leq n^2$, let  $ \lfloor k \rfloor _n $ denote $k$ modulo $n$. Then the above can be written more compactly as
\ba  E_k &  := & E_{\lfloor k \rfloor _n +1, k  - n \cdot \lfloor k \rfloor _n }  \nn \\
& = & e_{\lfloor k \rfloor _n +1} \otimes e _{k  - n \cdot \lfloor k \rfloor _n}  ^* ; \quad \quad 1 \leq k \leq n^2. \nn \ea
That is, given any matrix $A = [A _{ij} ]$ in the canonical basis $\{e_k \}$, if we re-label the entries of $A$ as:
$$ A = \bbm A_1 & A_2 & \cdots & A_n \\ A_{n+1} & A_{n+2} & \cdots & A _{2n} \\
\vdots &\ddots & & \vdots \nn \\
A_{(n-1)n +1} & \cdots & & A_{n^2} \ebm, $$ then
$$ A = \sum _{k=1} ^{n^2} A_k E_k. $$

\begin{lemma}
    Given $B \in \C^{n\times n}$, let $L_B , R_B \in \L (L^2 (\C ^{n\times n}, \trn) )$ be the operators of left and right multiplication by $B$.
Then with respect to the canonical orthonormal basis $\{ E_k \}$ of $L^2 (\trn)$
$$ [ L _B ] = B \otimes I \quad \quad \mbox{and} \quad [R _B] = I \otimes B ^T . $$
\end{lemma}

\begin{eg}
    If, for example, $n=2$ then
$$ [ L_A ] = \begin{bmatrix} A & 0 \\ 0 & A \end{bmatrix},$$ while
$$ [R_A ] = \bbm A_{11} & 0 & A _{21} & 0 \\
0 & A_{11} & 0 & A_{21} \\
A_{12} & 0 & A_{22} & 0 \\
0 & A_{12} & 0 & A _{22} \ebm. $$
\end{eg}

\begin{cor}
With respect to the canonical basis, $\{ E_k \} \subset L^2 (\A )$, if $q \in QC (\A), \  q \sim \{ q_k \} _{k=1} ^p$ then
$$ [T_q ] = \sum _{k=1} ^p  q_k \otimes \ov{q_k}. $$
\end{cor}

In the above $\ov{q_k}$ denotes the entry-wise complex conjugate of $q_k \in \C ^{n \times n}$.

Consider the linear transpose map $\mr{t} : \C ^{n\times n } \rightarrow \C ^{n\times n}$, as well as the anti-linear adjoint map $* : \C ^{n\times n} \rightarrow \C ^{n\times n}$.

Let $S : L^2 (\C ^{n\times n} ) \rightarrow L^2 (\C ^{n\times n} )$  be the unitary `tensor swap' with respect to the canonical basis:
$$ S (A \otimes B ) S  = B \otimes A. $$ If, for example, $n=2$, then $S E_1 = E_1$, $SE_2 = E _3$, $SE_3 = E_2$ and $SE_4 = E_4$;
$$ S = \bbm 1 & 0 & 0 & 0 \\
0 & 0 & 1 & 0 \\
0 & 1 & 0 & 0 \\
0 & 0 & 0 & 1 \ebm. $$
It is easy to verify that $S = T_{\mr{t}}$, the representing map of the transpose $\mr{t} : \C ^{n\times n } \rightarrow \C ^{n \times n}$. Similarly, the representing (anti-linear) map of $*$ is $T_* = S \circ \cc = \cc \circ S =:C$, the idempotent anti-linear isometry of entrywise complex conjugation  composed with the tensor swap with respect to the canonical basis. Recall that any anti-linear idempotent isometry is called a \emph{complex symmetry} or
a \emph{conjugation}.

\begin{cor}
    Let $q \in QC (\C ^{n\times n})$ be a unital quantum channel. Then $T_q$ commutes with the conjugation $C$:
    $$ C T_q = T_q C. $$
\end{cor}

\begin{remark}
One can construct examples to show that the statement in the above Corollary is necessary but not sufficient. Namely, one can find contractions, $T \in \L (L^2 (\A ) )$ so that $T I_\A = I_\A$ (a necessary requirement for $T$ to represent a unital quantum channel), such that $T$ commutes with the conjugation $C$, and yet $T \neq T_q$ for any $q \in QC (\A )$. Perhaps the simplest example is the transpose map, whose representing matrix is, as above, $T_{\mr{Tr}} = S$, the unitary tensor swap operator. Then $CT_{\mr{Tr}} = (\cc\circ S)S = \cc =S(S\circ \cc) = T_{\mr{Tr}} C$. It is well-known that the transpose map is not completely positive, and hence is not the representing contraction for any quantum channel.\\
More generally, $T_q$ commutes with $C$ if and only if $q(X)^* = q(X^*)$. This holds if and only if the matrix $C_q: = \sum_{i,j=1}^n E_{ij}\otimes q(E_{ij})$ is Hermitian. A famous theorem of Choi asserts that $q: \C^{n\times n}\rightarrow \C^{m\times m}$ is completely positive if and only if $C_q \geq 0$, so maps commuting with $C$ that are not completely positive stand in the same relation to completely positive maps as do Hermitian matrices to positive semidefinite matrices \cite{Choi}.
\end{remark}

\subsection{Dilations of contractions}

If $T$ is a contraction on a Hilbert space, $\H$, a unitary operator $U \in \L (\K )$, on a larger Hilbert space $\K \supseteq \H$ is called a $N$-dilation of $T$ if
$$ P_\H U^k | _\H = T^k; \quad \quad 1 \leq k \leq N. $$ If this holds for all $N \in \N$, $U$ is called a \emph{power dilation} of $T$.

Any contraction, $T$, always has a unitary $N$-dilation \cite{Sha2014Ndil}. A $1$-dilation is given by the \emph{Julia operator}:
$$ U := \bbm T & -\sqrt{I - TT^*} \\ \sqrt{I-T^* T} & T^* \ebm, $$ a $2-$dilation is given by
$$ U_2 := \bbm T & 0 & -\sqrt{TT^*} \\ \sqrt{I -T^*T} & 0 & 0 \\ 0 &1 & T^* \ebm, $$ and the pattern is now apparent with the $N$-dilation acting on $\H \otimes \C ^N$, $N$ copies of $\H$. A similar construction allows one to construct a power dilation of $T$ acting on $\H \otimes \ell ^2 (\Z )$ \cite[Chapter I]{NF}.

\begin{lemma}
    If $\alpha $ is a finite matrix $N$-dilation of $q \in QC (\A )$ acting on a unital matrix $*$-algebra $\A \otimes \B$ then $T_{\alpha }$ is a unitary $N$-dilation of $T_q$.
If $\alpha$ is a $*$-automorphic power dilation of $q \in QC (\A )$ acting on a finite von Neumann algebra $\A \otimes \B$, then $T_\alpha$ is a unitary power dilation of $T_q$.
\end{lemma}

\begin{remark}
Although we have assumed in this section that $\A$ is a finite unital matrix $*$-algebra, these facts extend without difficulty to the case where $\A$ is a finite von Neumann algebra with faithful tracial state $\tra$. In particular, for any $q \in QC (\A )$ one can define a representing contraction $T_q \in \L (L^2 (\A ))$, as in Lemma \ref{repcontract}.
\end{remark}

\begin{proof}
    Suppose that $\alpha$ acts on the larger unital matrix algebra $\A \otimes \B$. Then, for any $1 \leq M \leq N$, if $\Phi$ denotes the unique
$\tra \otimes \tr _B$-preserving conditional expectation of $\A \otimes \B$ onto $\A \otimes I_\B$,
$$ q ^{(M)} \otimes I _\B = \Phi \circ \alpha ^{(M)} \circ \Phi. $$ By Proposition \ref{represent}, $T_\alpha =: U$ is unitary and $T_\Phi = P$ is the projection of $L^2 (\A \otimes \B )$ onto $L^2 (\A  \otimes I_\B )$. It follows that
$$ T_q ^{M} \otimes I = P U ^M P; \quad \quad 1 \leq M \leq N. $$
\end{proof}

\begin{remark}{ (minimal dilations and uniqueness)}
A unitary power dilation $(U , \K )$ of a contraction $(T , \H)$ is called \emph{minimal} if
$$ \K := \bigvee _{k \in \Z } U^k \H. $$ Any contraction $T$, on $\H$ always has a minimal unitary power dilation (simply restrict any
unitary dilation $U$ on $\K '$ to the reducing subspace $\K$ defined above), and this is unique up unitary equivalence via a unitary which restricts to the
identity on $\H$ \cite[Chapter I]{NF} \cite[Theorem 4.3]{Paulsen}.

In the case of $N$-dilations, one can again define minimality. Minimal $N$-dilations of contractions obey weaker uniqueness properties and are generally non-unique \cite[Section 2]{Sha2014Ndil}.
\end{remark}

\begin{cor} \label{nopower}
    If $q \in QC (\A )$ is a unital quantum channel acting on a finite unital matrix algebra, $\A$, and $q$ is not a unitary $*$-automorphism, then there is no power dilation $\alpha$ of $q$ which acts on a finite unital matrix $*$-algebra.
\end{cor}

\begin{proof}
    By the previous lemma, if such an $\alpha$ existed then $U := T_\alpha$ would be a unitary power dilation of $T_q$ which acts on a finite dimensional Hilbert space $\K \supset \H := L^2 (\A )$. However, this would imply that $\H$ is \emph{semi-invariant} for $U$, \emph{i.e.} if $P := P_\H$, then
    $$ P U (I-P) U P = 0. $$ It is known that a subspace $\H \subset \K$ is semi-invariant for an operator $U$, if and only if $\H$ can be written as the direct difference of two invariant (or co-invariant) subspaces $\K _1 \subset \K _2 \subset \K$ for $U$:
    $$ \H = \K _2 \ominus \K _1, $$ \cite{Sarason-semi}. However, since $U$ is a finite dimensional unitary, any invariant subspace for $U$ is necessarily reducing, so that
    $\H$ is the direct difference of reducing subspaces for $U$, and $\H$ itself must be reducing for $U$. This proves that $T _q = U | _\H$ is unitary, and by Theorem \ref{represent}, $q$ is a unitary $*$-automorphism of $\A$.
\end{proof}

\begin{remark}
As discussed in Subsection \ref{DilFact}, any matrix factorizable $q \in QC (\A )$ is factorizable in the sense of \cite{Hup-factor,CAD-ergodic}, and hence has a $*$-automorphic power dilation in the sense of \cite{Kum-Markov} acting on a finite von Neumann algebra. The above corollary simply shows that this power dilation cannot act on a finite type-$I$ von Neumann algebra unless $q$ is already a unitary $*$-automorphism.
\end{remark}

\section{Outlook}
In this paper, we have shown that existence of matrix factorizations for a quantum channel $q$ is equivalent to existence of matrix N-dilations. As mentioned in Example \ref{Schureg}, the question of the existence of a factorization is related to the famous open Connes' embedding problem, which can be stated as the problem of deciding whether the existence of a factorization for a Schur product channel is equivalent to the existence of a matrix factorization for the same channel \cite{Hup-factor,dykema,ozawa}. So it is not unreasonable to believe that determining whether or not a given channel admits a matrix factorization, hence whether matrix dilations exist, is a hard problem in general. However, given that there are wide classes of channels for which the existence of factorizations can be guaranteed/excluded, an obvious avenue for further work is to find more necessary or sufficient conditions for the existence of matrix factorizations. \\
A related problem is the project of classifying the matrix algebras by which a given channel can be factorized. Consider the completely depolarizing channel on $2\times 2$ matrices:
$$q(X) = \tr(X) I_2.$$
There are (at least) two non-isomorphic matrix algebras through which $q$ factors: $\mathcal{N}_1 = \C^{2\times 2}$ and $\mathcal{N}_2 = (\C\oplus \C \oplus \C \oplus \C, \tr_4)$ where $\tr_4$ is the usual normalized trace on the space of $4\times 4$ matrices. To see this, let $U_1 = \sum_{i,j=1}^2 E_{ij}\otimes E_{ji}$. The blocks $u_{ij}$ are just the matrix units for $\C^{2\times 2}$, so $U\in \C^{2\times 2}\otimes \C^{2\times 2}$. In fact $U_1$ is the tensor-swap matrix which we have seen above: $U_1(x\otimes y) = y\otimes x$ for any $x,y\in \C^2$, and hence $U_1(A\otimes B)U_1^* = B\otimes A$ for any $A,B\in \C^{2\times 2}$. Hence $$\tr_B U(X\otimes I_2)U^* = \tr_B I_2 \otimes X = \tr(X) I_2 =q(X).$$
On the other hand, let $U_2 = \sum_{i=1}^4 \sigma_i \otimes E_{ii}$ for $E_{ii} \in \bigoplus_{i=1}^4 \C$, and $\{\sigma_i\}_{i=1}^4$ are the Pauli matrices. $U_2$ is unitary since each $\sigma_i$ is unitary, and $U_2 \in\C^{2\times 2}\otimes \mathcal{N}_2$. Finally,
$$\tr_B U(X\otimes I_4)U^* = \sum_{i=1}^4 \tr(E_{ii}) \sigma_i X\sigma_i^* = \frac{1}{4}\sum_{i=1}^4 \sigma_i X \sigma_i^*.$$
A simple matrix calculation confirms that this does indeed yield $q(X)$. \\
This example shows the potential difficulty in trying to classify which algebras can be used for the factorization of a given channel: we have two inequivalent factorizations by means of algebras with the same (minimal) dimension. A better understanding of which algebras can be used in factorizations is therefore an obvious direction for future work. \\

\paragraph{\bf Acknowledgements: \rm } JL and RTWM acknowledge support from the National Research Foundation of South Africa, NRF CPRR grant number $90551$.

\end{document}